\documentclass[10pt,final,twocolumn]{IEEEtran}
\usepackage{graphics,epsfig,theorem,latexsym,amssymb,amsmath,picinpar,psfrag,color,amsfonts, subfigure}
\usepackage{mathrsfs}
\usepackage{caption}
\usepackage{ulem}

\usepackage[utf8]{inputenc}
\usepackage{tabularx,ragged2e,booktabs,caption}
\newcolumntype{C}[1]{>{\Centering}m{#1}}

\newtheorem{proposition}{Proposition}

\newtheorem{lemma}{Lemma}

\newenvironment{proof}[1][Proof]{\begin{trivlist}
\item[\hskip \labelsep {\bfseries #1}]}{$\triangleleft$\end{trivlist}}

\newcommand{\Kv}{{K}}
\newcommand{\hax}{{\hat{x}}}
 
\newcommand{\RR}{\mathbb R}
\def\beeq#1{\begin{equation}{#1}\end{equation}}
\def\be{\begin{equation}}
\def\ee{\end{equation}}
\def\ben{\begin{equation*}}
\def\een{\end{equation*}}
\def\ba{\begin{array}}
\def\ea{\end{array}}
\def\bea{\begin{eqnarray}}
\def\eea{\end{eqnarray}}
\def\beann{\begin{eqnarray*}}
\def\eeann{\end{eqnarray*}}
\def\bmx{\begin{matrix}}
\def\emx{\end{matrix}}
\def\bpmx{\begin{pmatrix}}
\def\epmx{\end{pmatrix}}

\newcommand{\norm}[1]{\left| #1 \right|}

\def\dst{\displaystyle}
\newcommand{\dnorm}[1]{\Vert #1 \Vert}

\begin{document}

\title{A High-Gain Nonlinear Observer \\ with Limited Gain Power}

\author{Daniele Astolfi and Lorenzo Marconi 
\thanks{D. Astolfi is with CASY-DEI, University of Bologna,
Italy and with MINES ParisTech, PSL Research University, CAS, Paris, France, (daniele.astolfi@unibo.it).}
\thanks{L. Marconi is with the CASY-DEI, University of Bologna,
Italy, (lorenzo.marconi@unibo.it).}
\thanks{Work supported by the European project SHERPA  (ICT 600958).}}

\maketitle

\begin{abstract}
In this note we deal with a new observer for nonlinear systems of dimension 
$n$ in canonical observability form.
We follow the standard high-gain paradigm, but instead of having an observer of dimension $n$ with
a gain that grows up to power $n$, we design an observer of dimension 
$2n-2$ with a gain that grows up only to power 2. 
\end{abstract}

\begin{IEEEkeywords}
Observability, nonlinear observers, high-gain observers.
\end{IEEEkeywords}

\section{Introduction}
In this note we consider the problem of state observation for nonlinear systems of the form
\be\label{sys_normal}
\dot z = f(z) + \bar d(t)\;, \qquad y = h(z) + \nu(t)
\ee
where $z \in {\cal Z} \subseteq \RR^n$ is the state, $y \in \RR$ is the measured output,
$f(\cdot)$ and $h(\cdot)$ are sufficiently smooth functions, $\bar d(t) \in \RR^n$ is a  bounded disturbance and $\nu\in \RR$ is the measurement noise. 
Among the different techniques for observer design available in literature (see \cite{GK}, \cite{Besancon}) we are particularly interested to the so-called high-gain 
methods that have been shown to be effective in many control scenarios. In this respect we assume that the pair $(f(\cdot), h(\cdot))$ fulfils an {\it uniform observability assumption} (see Definition 1.2 in \cite{GK}), which implies the existence of a diffeomorphism $\phi: {\cal Z} \to \RR^n$ such that the dynamic of the new state variable $x= \phi(z)$ is described by the {\it
canonical observability form}  (see Theorem 4.1 in \cite{GK})
\be \label{sys}
\dot x = A_nx + B_n\varphi(x) + d(x,t) \,,\qquad y = C_n x + \nu(t)\;,
\ee
where  $\varphi(\cdot)$ is a locally Lipschitz function, 
\[ d(x,t) := \left . {d\phi(z) \over dz} \right |_{z = \phi_{\rm inv}(x)} \bar d(t)\,,
\] 
 with $\phi_{\rm inv}(\cdot)$ the inverse of $\phi(\cdot)$ (namely $\phi_{\rm inv}\circ \phi(z)=z$ for all $z \in {\cal Z}$), 
 and $(A_n,B_n,C_n)$ is a triplet in ``prime form'' of dimension $n$, that is
\be \label{defABC}
\begin{array}{rcl}
A_n &=& \left( \ba{cc} {0}_{(n-1)\times 1} & I_{n-1} \\
0 & {0}_{1 \times (n-1)} \ea \right )\!,
\;\; B_n = \left( \ba{c} { 0}_{(n-1)\times 1} \\ 1 \ea \right )\!,\\
C_n &=& \left( \ba{cc}  1 & { 0}_{1 \times (n-1)} \ea \right ) \;.
\end{array}
\ee
System (\ref{sys}) is defined on a set ${\cal X} := \phi({\cal Z}) \subseteq \RR^n$.

For the class of systems (\ref{sys}), it is a well-known fact  (\cite{Bornard1991})  that the problem of asymptotically, in case $d(x,t) \equiv 0$ and $\nu(t)\equiv 0$,  estimating the state $x$ can be addressed by means of  a {\it high-gain} nonlinear observer of the form
\be\label{standardhgo}
\dot \hax =  A_n\, \hax + B_n\, \varphi_s(\hax) + D_n(\ell)\, \Kv_n \, \left(C_n\hax - y\right)
\ee
with 
$$
D_n(\ell) = {\rm diag}\, \Big(\,\ell\,,\;\ldots\;,\;\ell^n\,\Big) \;,\qquad 
K_n = \Big(\,c_1\;\;\cdots\;\;c_n\,\Big) ^\top\;,
$$
where $\ell$ is a high-gain design parameter taken sufficiently large (i.e. $\ell \geq \ell^\star$ with $\ell^\star\geq 1$), the $c_i$'s are chosen so that the matrix $(A+KC)$ is Hurwitz (i.e. all its eigenvalues are on the left-half complex plane),  and $\varphi_s(\cdot): \RR^n \to \RR$ is an appropriate "saturated" version of $\varphi(\cdot)$. 
As a matter of fact, it can be proved that if $d(x,t) \equiv 0$ and $\nu(t)\equiv 0$,  if $\varphi(\cdot)$ is uniformly Lipschitz in ${\cal X}$, namely there exists a $\bar \varphi >0$ such that 
\beeq{\label{unifLips}
\|\varphi(x') - \varphi(x'')\| \leq \bar \varphi \|x'-x''\| \quad \forall \, x',x'' \in {\cal X}\,,
}
and $\varphi_s(\cdot)$ is chosen bounded and to agree with $\varphi(\cdot)$ on $\cal X$,  
the observation error $e(t)=x(t) - \hat x(t)$ originating from (\ref{sys}) and (\ref{standardhgo}) exponentially converges to the  origin with an exponential decay rate of the form 
\[
 \|x(t) - \hat x(t)\| \leq \alpha \, \ell^{n-1}\,\mbox{exp}(-{\beta \ell} t )\|x(0)-\hat x(0)\|\;,
\]
where $\alpha$ and $\beta$ are positive constants,  for all possible initial condition $\hat x(0)\in \RR^n$ {\it as long as}  $x(t) \in \cal X$.  In particular, note that the exponential decay rate
 may be arbitrarily assigned by the value of $\ell$ with a polynomial "peaking" in $\ell$ of order $n-1$.  It is worth noting that the uniform Lipschitz condition 
 (\ref{unifLips}) is automatically fulfilled if $\cal X$ is a compact set. 
In case $d(x,t)$ or $\nu(t)$ are not identically zero,  as long as they are  bounded for all $t\geq 0$ and for all\footnote{boundedness of $d(x,t)$ is automatically guaranteed if $\cal X$ is compact. This, in turn, is the typical case when such observers are used in  semiglobal output feedback stabilisation problems, \cite{TeelPraly1994}.} $x \in \cal X$, the observer (\ref{standardhgo}) guarantees a 
bound on the estimation error that depends on the bound of $d(\cdot, \cdot)$, of $\nu(\cdot)$ and on the value of $\ell$. In particular, the following asymptotic bounds can be proved 
\begin{multline*}
 \lim_{t \to \infty} \sup \| x(t) - \hat x(t)\| \, \leq\\
 \gamma \, \max\left\{\,  \lim_{t \to \infty} \sup \|{1\over \ell} \, \Gamma(\ell)\, d(x(t),t)\|\,,
 \lim_{t \to \infty} \sup \|\ell^{n-1} \, \nu(t)\|\, \right\}
\end{multline*}
 where $\gamma$ is a positive constant and 
 \beeq{\label{Gamma}
 \Gamma(\ell)= \mbox{\rm diag}(\ell^{n-1}, \, \ldots, \ell, \,1)\,.
 }
  As above the previous asymptotic bound holds for all possible $\hat x(0) \in \RR^n$ as long as $x(t) \in \cal X$. 
Note that a high value of $\ell$ leads to an arbitrarily small asymptotic gain on the $n$-th component of the disturbance $d(\cdot,\cdot)$. On the other hand, a large value of $\ell$ is, in general, detrimental for the sensitivity of the asymptotic estimate to the sensor noise and to the first $n-2$ disturbance components.

Observers of the form (\ref{standardhgo}) are routinely used in many observation and control problems. For instance, the feature   of having an exponential decay rate and an asymptotic bound on the last component of $d$ that can be arbitrarily imposed by the value of $\ell$ is the main reason why the above observer plays a fundamental role in output feedback stabilisation and in setting up semiglobal nonlinear separation principles
 (\cite{TeelPraly1994},
 \cite{Atassi}). In that case the set $\cal X$ is an arbitrarily large compact set which is made invariant by the design of the state feedback stabilisation law and of the high-gain observer.   We observe that, although the asymptotic gain with respect to $\nu$ increase with $\ell^{n-1}$, the observer is anyway able to guarantee ISS with respect to the sensor noise
  (\cite{KhalilPraly}).
The main drawback of observers of the form (\ref{standardhgo}), though, is related to the increasing power (up to the order $n$) of the high-gain parameter $\ell$, which makes the practical numerical implementation an hard task when $n$ or $\ell$ are very large. 
Motivated by these considerations, in this note we propose a new observer for the class of systems (\ref{sys}) that preserves the same high-gain features of (\ref{standardhgo}) but which substantially overtakes the implementation problems due to the high-gain powered up to the order $n$.
Specifically, we present a high-gain observer structure with a gain which grows only up to power 2 (instead of $n$), at the price of having
the observer state dimension $2n-2$ instead of $n$.

\section{Main Result}

We start by presenting a technical lemma instrumental to the proof of the main result presented in Proposition \ref{MainResult}.
Let $E_i\in \RR^{2 \times 2}$, $Q_i\in \RR^{2 \times 2}$, $i=1, \ldots, n-1$, and $N\in \RR^{2\times 2}$  be matrices defined as  
$$
E_i= \left(\begin{array}{cc}
-k_{i1} & 1
\\
-k_{i 2} & 0
\end{array}\right)
\! ,\;
Q_i= \left(\begin{array}{cc}
0 & k_{i1}
\\
0 & k_{i 2}
\end{array}\right)
\!,\;
N= \left(
\begin{array}{cc}
0 & 0
\\
0 & 1
\end{array}\right)
$$
where $(k_{i1},k_{i2})$ are positive coefficients,
and let $M \in \RR^{(2n-2) \times (2n-2)}$ be  the block-tridiagonal matrix defined as
\be\label{Mdef}
M\;=\; \left(
\begin{array}{ccccccc}
E_1 &  N & 0&&\ldots& \ldots &0
\\
Q_2 &E_2 & N &\ddots &&&\vdots
\\
0&\ddots &\ddots &
\ddots &\ddots &&\vdots
\\
\vdots&\ddots& Q_i &E_i & N &\ddots &\vdots
\\
\vdots&&\ddots&\ddots &\ddots &
\ddots & 0
\\
\vdots&&&\ddots& Q_{n-2}& E_{n-2}
& N
\\
0&\ldots&\ldots&\ldots&0& Q_{n-1}& E_{n-1}
\end{array}\right) \;.
\ee
It turns out that the eigenvalues of $M$ can be arbitrarily assigned by appropriately choosing the coefficients $(k_{i1},k_{i2})$, $i=1,\ldots,n-1$, as claimed in the next lemma. 
 
\begin{lemma}\label{Lemma1}
Let ${\cal P}(\lambda) = \lambda^{2n-2}+ m_1 \lambda^{2n-3} + ... + m_{2n-3}\lambda + m_{2n-2}$ be an arbitrary Hurwitz polynomial. There exists a choice of $(k_{i1},k_{i2})$, $i=1,\ldots, n-1$, such that the characteristic polynomial of $M$ coincides with ${\cal P}(\lambda)$. 
\end{lemma}
The proof of this Lemma is deferred to the appendix where a {\it constructive} procedure for designing $(k_{i1},k_{i2})$ given ${\cal P}(\lambda)$ is presented.\\[1mm]
The structure of the proposed observer has the following form  
\be\label{observer}
\begin{array}{rcl}
\dot \xi_i &=& A \,\xi_i + N \,  \xi_{i+1} + D_2(\ell) \, K_i \, e_i \qquad i=1,\ldots, n\!-\!2\\
&\vdots&\\
\dot  \xi_{n-1} &=& A \, \xi_{n-1} + B \, \varphi_s(\hax') +D_2(\ell) \, K_{(n-1)} \,e_{n-1} 
\end{array}
\ee
where $(A,B,C)$ is a triplet in prime form of dimension $2$, $\xi_i \in \RR^2$, $K_{i} = \left(k_{i1} \; k_{i2} \right )^T$, $D_2(\ell) = \mbox{diag}(\ell, \, \ell^2)$, 
 \[
  \hax' = L_1 \,  \xi  \qquad L_1 = \mbox{blkdiag}\ (\underbrace{C , \;  \ldots,\; C}_{(n-2) \; \mbox{times}}, I_2   )\]
  \[ \xi = \mbox{col}( \xi_1,\ldots,  \xi_{n-1})  \in \RR^{2n-2}  \,,
 \]
\[
  e_1 =  y - C \xi_1\, ,\qquad e_i = B^T  \xi_{i-1}- C \xi_i \qquad i=2, \ldots, n-1\,,
\]
 and $\varphi_s(\cdot)$ is an appropriate saturated version of $\varphi(\cdot)$. 

The variable $\hax'$ represents an asymptotic estimate of the state $x$ of (\ref{sys}). It is obtained by ``extracting" $n$ components from the state $\xi$ according to the matrix $L_1$ defined above.   As clarified next, the redundancy of the observer can be used to extract from $\xi$ an extra state estimation that is   
 \[
  \hat x'' = L_2 \, \xi  \qquad L_2 = \mbox{blkdiag}\ (I_2\,, \; \underbrace{B^T , \;  \ldots,\; B^T}_{(n-2) \; \mbox{times}} )\,.
  \]
 The following proposition shows that the observer (\ref{observer}) recovers the same asymptotic properties for the two estimates $\hat x'$ and $\hat x''$ of the ``standard" high-gain observer (\ref{standardhgo}). 
In the statement of the proposition we let 
 \[
  {\hat {\bf x}} = \mbox{col}(\hat x', \, \hat x'')\,,
  \qquad 
   {\bf x} = \mbox{col}(x, \, x)\,.
 \]

\begin{proposition}\label{MainResult}\label{Sec:Main}
Consider system \eqref{sys}
 and the observer \eqref{observer} with the coefficients $\left(k_{i1} \; k_{i2} \right )$ fixed so that the matrix $M$ defined in \eqref{Mdef} is Hurwitz (see Lemma \ref{Lemma1}).
  Let $\varphi_s(\cdot)$ be any bounded function that agrees with $\varphi(\cdot)$  on $\cal X$, and assume that $d(x,t)$ is bounded for all $x \in \cal X$ and for all 
 $t\geq 0$.  Then there exist $c_i>0$, $i=1,\ldots,4$, and $\ell^\star\geq1$ such that for any $\ell \geq \ell^\star$ and for any $\xi(0) \in \RR^{2n-2}$,
   the following bound holds
\begin{multline}\label{prop1}
 \| \hat {\bf x}(t) -  {\bf x}(t) \| \, \leq \, \max\{c_1\, \ell^{n-1} \,\mbox{\rm exp}(-{c_2  \ell}\, t) \, \|\hat {\bf x}(0) -  {\bf x}(0) \|\,, \\
  {c_3} \,\|{1\over \ell} \, \Gamma(\ell)\, d(\cdot)\|_\infty, \; {c_4} \, \ell^{n-1}\|\,\nu(\cdot)\|_\infty \}
\end{multline}
where $\Gamma(\ell)$ is as in \eqref{Gamma},  for all $t\geq 0$ such that $x(t) \in \cal X$.  
\end{proposition}

\begin{proof}

Consider the change of variables
\[
\xi_i \; \mapsto \; \tilde \xi_i := \xi_i - \mbox{col}\left(x_i,\; x_{i+1} \right ) \qquad
i=1,\ldots, n-1\;,
\]
by which system (\ref{observer}) transforms as 
\[
\ba{rcl}
\dot {\tilde \xi}_1 &=& H_1 \, \tilde \xi_1 + N \,\tilde \xi_2- \bar d_1(t) + D_2(\ell)K_1 \nu (t)\\[2mm]
\dot{\tilde \xi}_i &=& H_i \, \tilde \xi_i + N \tilde \xi_{i+1} + D_2(\ell) \, K_i \, B^T \, \tilde \xi_{i-1} - \bar d_i(t) \\
&& \hspace*{4.5cm} i=2, \ldots, n-2\\[2mm]
\dot {\tilde \xi}_{n-1} &=& H_{n-1} \, \tilde \xi_{n-1} + D_2(\ell) \, K_{n-1} \, B^T \, \tilde \xi_{n-2} + \\[1em]
&&\hspace*{3.7cm} B \, \Delta \varphi(\tilde \xi,x) - \bar d_{n-1}(t)
\ea
\]
where 
$$H_i = A - D_2(\ell) \,K_i \,C\,, \quad   \bar d_i(t)=(d_i(x(t), t),d_{i+1}(x(t), t))^T\,,  $$ 
for $i=1,\ldots, n-1$,
 with $d_i(t)$  the $i$-th element of the vector $d(t)$, 
 $\tilde \xi = \mbox{col}(\tilde \xi_1, \ldots,  \tilde \xi_{n-1} )$, and 
\[
\Delta \varphi(\tilde \xi,x) =  \varphi_s(L_1 \tilde \xi + x)  - \varphi(x)\;.
\]
Rescale now the variables $\tilde \xi_i$ as follows 
\[
 \tilde \xi_i \; \mapsto \;    \varepsilon_i := \ell^{2-i} \, D_2(\ell)^{-1} \, \tilde \xi_i \qquad i=1,\ldots,n-1\,.
\]
By letting $\varepsilon = \mbox{col} (\varepsilon_1 \; \ldots \; \varepsilon_{n-1})$, an easy calculation shows that 
\beeq{\label{varepsilon}
 \dot \varepsilon = \ell M \varepsilon + \ell^{-(n-1)} \,\left (  B_{2n-2} \Delta \varphi_\ell(\varepsilon,x) +\upsilon_{\ell}(t) 
+  n_\ell(t)
 \right )\;,
}
where $B_{2n-2}$ is the zero column vector of dimension $2n-2$ with a $1$ in the last position, and 
\begin{multline*}
\upsilon_{\ell}(t) = - \mbox{col}\left (\ell^{n}D_2(\ell)^{-1} \bar d_1,\;\ldots , \; \ell^{n-i+1}D_2(\ell)^{-1} \bar d_i,\; \ldots, \; \right.\\
\left.\ell^2D_2(\ell)^{-1} \bar d_{n-1} \right  )\;,
\end{multline*}
$$
 n_\ell (t) = \ell^n \bar K_1 \nu (t) \;, \qquad \bar K_1 = \mbox{col}(K_1, 0, \ldots, 0)\;,
$$
and 
$$
\Delta \varphi_\ell(\varepsilon,x)=  \varphi_s(L_1\, S(\ell) \varepsilon + x)  - \varphi(x)\;,
$$
where $S(\ell) = \mbox{diag}({1 \over \ell}\,, \; 1\,, \; \ell\,, \; \ldots, \ell^{n-3} ) \otimes D_2(\ell)$. 
  Being $\varphi(\cdot)$ uniformly Lipschitz in $\cal X$ and $\varphi_s(\cdot)$ bounded, 
there exists a $\delta_1> 0$, $\delta_2>0$ and $\delta_3>0$ such that
\[\ba{rcl}
\|\ell^{-(n-1)} \, \Delta \varphi_\ell(\varepsilon,x)\| & \leq & \delta_1 \, \|\varepsilon\|\,, \\[.5em]
 \| \ell^{-(n-1)} \upsilon_{\ell}(t) \|  & \leq & \delta_2 \, \|\ell \, D_n(\ell)^{-1}\, d(x(t), t)  \| \;, 
 \\[.5em]
\dnorm{\ell^{-(n-1)}  n_\ell (t)} & \leq  &\delta_3 \dnorm {\ell \,\nu(t)} \;,
\ea\]
for all $\varepsilon \in \RR^{2n-2}$, $x \in \cal X$ and $\ell \geq 1$.  The rest of the proof follows standard Lyapunov arguments that, for sake of completeness, are briefly recalled. Let $P=P^T$ be such that $P M + M^T P = -I$ and consider the Lyapunov function $V= \varepsilon^T P \varepsilon$.  Taking derivative of $V$ along the solutions of (\ref{varepsilon}), using the previous bounds and letting $\ell^\star = 4 \, \delta_1 \| P\|$, one obtains that there exist  positive constants $a_1$, $a_2$, $a_3$ such that for any $\ell \geq \ell^\star$
\begin{multline*}
 \| \varepsilon \| \geq \max\left\{a_1 \|D_n(\ell)^{-1}\,  d(x(t), t)  \|, a_2 \dnorm{\nu(t)} \right\}
 \qquad \Rightarrow  \\[.5em]
 \qquad \dot V \leq -a_3 \, \ell \, \|\varepsilon\|^2 \,.
\end{multline*}
As $P$ is symmetric and definite positive, it turns out that $\underline{\lambda} \|\varepsilon\|^2\leq V(\varepsilon) \leq \bar{\lambda} \|\varepsilon\|^2$ where $\underline{\lambda}$ and $\bar{\lambda}$ are respectively the smallest and the highest eigenvalue of $P$. By using these bounds, the previous implication  leads to conclude that, as long as 
$$
V \geq \bar{\lambda }\, \max\left\{ a_1^2 \, \|D_n(\ell)^{-1}\, d(x(t), t) \|^2\;, a_2^2 \,\dnorm{\nu(t)} \right\}
$$
then $V(t) \leq \mbox{exp}( -{a_3 \, \ell \over \bar{\lambda}} t) V(0)$. By using again the bound on $V$ in terms  of $\underline{\lambda}$ and $\bar{\lambda}$, the following estimate on $\varepsilon(t)$ can be easily obtained 
\begin{multline*}
 \|\varepsilon(t)\| \leq \max\{ \, a_4 \, \mbox{\rm exp} (- a_5 \, \ell \, t ) \, \|\varepsilon(0)\|\,, \; \\ a_6 \, \|D_n(\ell)^{-1}\, d(x(t), t)  \| \,, 
\; a_7\dnorm {\nu(t)}\} 
\end{multline*}
where $a_4 = \sqrt{\bar{\lambda}/ \underline{\lambda}}$, $a_5 = a_3 /2 \bar{\lambda}$, $a_6=\sqrt{\bar{\lambda}/ \underline{\lambda}} \, a_1$, 
$a_7=\sqrt{\bar{\lambda}/ \underline{\lambda}} \, a_2$. 
Now, using the fact that, for all  $\ell>1$, ${\ell^{-(n-1)}} \,\| \tilde \xi\| \leq \|\varepsilon\| \leq \| \tilde \xi\|$, the previous bound  leads to the following estimate on $\tilde \xi(t)$
\begin{multline*}
 \|\tilde \xi(t)\| \leq \max\{ \, a_4 \,\ell^{n-1} \mbox{\rm exp} (- a_5 \, \ell \, t ) \, \|\tilde \xi(0)\|\,, \;\\ a_6 \,\ell^{n-1} \,\|D_n(\ell)^{-1}\, d(x(t), t)  \| \,, 
 \;  a_7 \, \ell^{n-1}\dnorm{\nu(t)} \} 
\end{multline*}
by which the claim of the proposition immediately follows by bearing in mind the definition of $\Gamma(\ell)$, $\hat {\bf x}$, $\bf x$, and by noting that $\|\tilde \xi\| \leq \|\hat {\bf x} - {\bf x} \| \leq 2 \|\tilde \xi\|$.\,\end{proof}



\section{About the sensitivity of the observer to high frequency noise in the linear case}
\label{sec:Sensitivity}
The trade-off between the speed of the state estimation and 
the sensitivity to measurement noise is a well-known fact in the observer theory. In this respect, high-gain observers tuned to obtain fast estimation dynamics are necessarily very sensitive to high-frequency noise.  Bounds on the estimation error in presence of measurement noise 
for the standard high-gain observers have been studied, for instance, in  \cite{Vasil} and \cite{Ball}, and
different techniques have been developed in order to improve rejection, mainly based on gain adaptation (see, among others,  \cite{Ahrens}, \cite{Sanfelice}).  

In this section we compare the properties of the standard high-gain observer \eqref{standardhgo} and the proposed observer \eqref{observer} 
with respect to high-frequency measurement noise by specialising the analysis to linear systems.  

In particular we consider systems of the form \eqref{sys} with $\varphi(\cdot)$ a linear function of the form $\varphi(x) = \Phi \, x$ where $\Phi$ is a row vector of dimension $n$.
Moreover, in this contest, we consider $d(t)\equiv 0$ and 
\be\label{noise}
\nu(t) = a_N \sin(\omega_N \,t + f_N)
\ee
where $a_N>0$, $\omega_N>0$, and $f_N$ are constants. It is shown that the ratio between the asymptotic estimation error on the $i$-th   state variable provided by the new observer \eqref{observer}  and the one provided by the standard observer \eqref{standardhgo} is a strictly decreasing polynomial function of the noise frequency for $i=2,\ldots, n$. In this regard the new observer has better asymptotic properties with respect to high-frequency noise as far as the state estimation variables are concerned (except for the first one). This is formalised in the next proposition.

\begin{proposition}
Let $n \geq 1$. Let $K_n\in {\mathbb R}^n$,  $K_i\in  {\mathbb R}^2$, $i=1,\ldots, n-1$ and $\ell \in  {\mathbb R}$ be fixed so that the error dynamics of the observers \eqref{standardhgo} and \eqref{observer} are Hurwitz.
Moreover, with $\rho \in \{1,2,\ldots, n\}$  denoting the position index of the first non zero coefficient of the vector $\Phi$ (and $\rho=n$ if $\Phi$ is the zero vector), let $r_i' \geq 1$ be the constants defined as
 $$
r_i' = \min\,\{\;i\;,\; (n-1)\;,\;(\rho+n-i+1)\} \qquad \quad i=1,\ldots,n\;.
$$
There exist $\omega_N^\star>0$ and $\bar c_i>0$ such that for all $\omega_N > \omega_N^\star$, $a_N>0$ and $f_N$ we have
$$
\dfrac{\dst \limsup_{t\to+\infty} \norm{\hax'_i(t)-x_i(t)}}{\dst \limsup_{t\to+\infty} \norm{\hax_i(t)-x_i(t)}} \; \leq \; 
\bar c_i\; \omega_N^{-(r_i'-1)}
\qquad \forall \; i = 1, \ldots, n.
$$
\end{proposition}
\begin{proof}
Consider system \eqref{sys} and the standard high-gain observer \eqref{standardhgo}. By letting  $e = \ell \, D_n(\ell)^{-1} \;(\hat x - x)$, the $e$-dynamics  read as
\be\label{pr1}
\dot e = \ell \big(A_n + K_n C_n\big) \, e + B_n \Phi \,\Theta_n(\ell)\,  e + \ell K_n \nu(t) 
\ee
where  $\Theta_n(\ell) = \dfrac{1}{\ell^n}\,D_n(\ell)$ with $D_n(\ell)$ defined in \eqref{standardhgo}. It is a linear system that is Hurwitz by the choices of $K_n$ and $\ell$.
Similarly, consider system \eqref{sys} and the new observer \eqref{observer}. With $\varepsilon$ defined as in the proof of Proposition 1 we  obtain system \eqref{varepsilon} compactly rewritten as
\be\label{pr2}
\dot \varepsilon = \ell M\, \varepsilon +
B_{2n-2}\Phi\, \Theta_n(\ell)L_1 \, \varepsilon +
\ell \bar K_1 \nu(t)\;.
\ee

It is an Hurwitz system by the choices of  $K_i$, $i=1,\ldots,  n-1$, and $\ell$. 
We consider now the $n$ systems given by the dynamics (\ref{pr1}) with input $\nu$ and with outputs 
 \[
 \hat x_i - x_i = \ell^{i-1} e_i\,, \qquad i =1,\ldots n\,, 
\]
and we denote by $F_i:{\mathbb R} \to {\mathbb C}$, $i=1,\ldots,n$, the harmonic transfer functions of these systems. A simple computation shows that these systems have relative degree $r_i=1$, for all $i=1,\ldots, n$. 
Similarly, we consider the $n$ systems given by the dynamics (\ref{pr2}) with input $\nu$ and outputs
\[
\ba{rcl}
 \hat x_i' - x_i &=& \ell^{i-1} \varepsilon_{i,1} \qquad i =1,\ldots n-1,\\
  \hat x_n' - x_n &=& \ell^{n-1} \varepsilon_{n-1,2}\,,
\ea
\]  
and we denote by $F'_i :{\mathbb R} \to {\mathbb C}$, $i=1,\ldots,n$,  the harmonic transfer functions of these systems. 
Simple computations show that these systems have relative degree $r_i'$, for all $i=1,\ldots, n$, with $r'_i$ defined in the statement of the proposition.
By definition of harmonic transfer function and by the fact that systems (\ref{pr1}) and (\ref{pr2}) are Hurwitz, it turns out that 
\[\ba{rcl}\displaystyle
\limsup_{t\to+\infty}\norm{\hax_i(t) - x_i(t)}&  = &\norm{F_i(\omega_N)} a_N\\
\displaystyle
\limsup_{t\to+\infty}\norm{\hax'_i(t) - x_i(t)}&  = &\norm{F'_i(\omega_N)} a_N
\ea
\]
for any $\omega_N\geq0$ and $\forall \,i=1,\ldots,n$.
Furthermore, by the fact that $F_i(\omega)$ and $F_i'(\omega)$ have, respectively, relative degrees $r_i=1$ and $r_i'$, $i=\,\ldots, n$, it turns out that  there exist positive $c_i$, $c_i'$ and $\omega_N^\star>0$ such that
\[
\ba{rcl}
|F_i(\omega)|  &\geq&  c_i \, \omega^{-1}\\
|F_i'(\omega)|   &\leq& c_i' \, \omega^{-r_i'}
  \ea
  \qquad \forall \, \omega \geq \omega_N^\star\,,
\]
by which the result immediately follows.
\end{proof}

\section{Example: Observer for the uncertain Van Der Pol Oscillator}
Let consider the uncertain Van der Pol oscillator
\be\label{eq:Vanderpol}
\ddot z = -\alpha^2 z + \beta (1 - z^2) \dot z  \;, \qquad \qquad y = z + \nu(t) \;,
\ee
where $y\in \RR$ is the measured output and $\alpha, \beta$ are  uncertain
constant parameters. We let $\mu = (\alpha^2,\beta)^\top$ and we assume that $\mu \in {\cal U}$, with ${\cal U}$ a compact set of $\RR^2$ not containing the origin. The state $(z,\dot z)$ belongs to a compact invariant set ${\cal W} \subset \RR^2$, which is the limit cycle of the Van Der Pol oscillator. We observe that $\cal W$ depends on $\mu$. Following  \cite{AstolfiNOLCOS}, system (\ref{eq:Vanderpol}) extended with $\dot \mu=0$ can be immersed  into a system in the canonical observability form \eqref{sys} with $n=5$. As a matter of fact,  let $z_{[i,j]} := {\rm col}\;(z^{(i)},\ldots,z^{(j)})$ be the vector of time derivatives of $z$, with $0\leq i < j$, and let  ${\cal X}$  be the compact set of $\RR^5$ such that $z_{[0,1]} \in \cal W$ $\Rightarrow$ $z_{[0,4]} \in  {\cal X}  $. Simple computations show that 
$$
z_{[2,4]} = \Upsilon(z_{[0,3]}) \;\mu\;,\qquad
z^{(5)} = \rho(z_{[0,4]}) \; \mu
$$
where 
$$
\Upsilon(z_{[0,3]})  = \bpmx 
-z & (1-z^{2})\dot z \\
-\dot z & \ddot z - 2y\dot z^{2} - z^{2} \ddot z \\
-\ddot z & z^{(3)}-2 \dot z^{3} -6 z \dot z \ddot z - z^{2} z^{(3)}
\epmx
$$
and 
$$
 \rho(z_{[0,4]}) =  \bpmx z^{(3)}\;,\;
z^{(4)}(1-z^2) - 12 \dot z^{2}  \ddot z - 6 z  \ddot z^2 - 8 z  \dot z z^{(3)}
\epmx
$$
with (see \cite{Forte})
\[
 \mbox{rank} \Upsilon(z_{[0,3]}) = 2 \qquad \mbox{for all } \;   z_{[0,4]} \in  {\cal X} \,.
\]
Hence, by letting 
\[
 \varphi(z_{[0,4]}) =    \rho(z_{[0,4]}) \; \hat \mu(z_{[0,3]})
\]
with
\[
\hat \mu(z_{[0,4]} ) =  \Upsilon^\dag (z_{[0,4]}) \; z_{[2,4]}\,,
\]
where $\Upsilon^\dag (\cdot)$ is the left-inverse of $\Upsilon(\cdot)$, and letting $x =z_{[0,4]}$, it is immediately seen that system (\ref{eq:Vanderpol}) and $\dot \mu =0$ restricted to ${\cal W} \times \cal U$ is immersed into the system
\be\label{eq:van_prime}
\dot x \;= \;A_5 \,x + B_5\,\varphi(x) \;,\qquad y \;= \;C_5 \,x + \nu(t)\;,
\ee
where $A_5,B_5,C_5$ is a triplet in prime form of dimension $5$.

By following the prescriptions of Section II, we implemented the proposed observer \eqref{observer} as 
\begin{equation}
\label{eq:observer}
\ba{rcl}
\dot {\xi}_1 & = & A \, \xi_1 + N \,\xi_2 + D_2(\ell_1)K_1 (y - C\xi_1) \\
& \vdots & \\
\dot {\xi}_4 & = & A \, \xi_4 + B\, \varphi_s(\hat x')  + D_2(\ell_1)K_4 (B^\top \xi_3 - C\xi_4)
\ea\end{equation}
where  
 $\varphi_s(\cdot)$ is any locally Lipschitz bounded function that agrees with $\varphi(\cdot)$ on $\cal X$, $\xi_i = (\xi_{i1}, \xi_{i2})$ and
the coefficients of $K_i$  are $k_{11} = 0.6$, $k_{12} = 0.3$, $k_{21} = 0.6$, $k_{22} = 0.111$, $k_{31} = 0.6$, $k_{32} = 0.0485$, 
$k_{41} = 0.6$, $k_{42} = 0.0178$, 
such that the roots of ${\cal P}(\lambda)$ are
$-0.1$, $-0.2$, $-0.2$, $-0.3$, $-0.3$, $-0.4$, $-0.4$, $-0.5$. With the same notation
 of \eqref{observer} we have $\xi = (\xi_1,\xi_2,\xi_3,\xi_4)$,
 $\hax'= L_1 \xi$ and $\hax'' = L_2 \xi$.
 
In the simulations we fixed $\alpha = 1$, $\beta = 0.5$, gain $\ell_1=100$ and initial conditions 
$(z,\dot z)=(1,0)$ for \eqref{eq:Vanderpol} and 
 $\xi = 0$ for \eqref{eq:observer}. 
Figures \ref{graph1} and \ref{graph2} show the error state estimate $\norm{\hax'-x}$ and 
$\norm{\hax''-x}$ of the proposed observer \eqref{eq:observer}
for the first two components (namely the estimation of the state of (\ref{eq:Vanderpol})) when there is not sensor noise.

By following Section III we compared the observer \eqref{eq:observer} with a standard
high-gain observer in presence of high-frequency sensor noise, numerically taken as $\nu(t) = 10^{-2}\sin(10^3 \,t)$. The high-gain observer has been implemented as
\be\label{eq:shgo}
\dot \hax = A_5 \hax + B_5 \varphi_s(\hax)  + D(\ell_2) K_5 (y - C_5 \hax) 
\ee
where   $K_5 = (1.5, \; 0.85, \; 0.225, \; 0.0274, \; 0.0012)^T$
so that the eigenvalues of $(A_5+K_5C_5)$ are
$-0.1$, $-0.2$, $-0.3$, $-0.4$, $-0.5$, and $\ell_2=100$.
 Table 1 shows the normalized asymptotic error magnitudes of the proposed observer 
\eqref{eq:observer}
 and the standard high-gain observer \eqref{eq:shgo}, where the normalized asymptotic error for the $i$-th estimate is defined as
$
\|\hat x_i - x_i \|_a =  \displaystyle \lim_{t\to\infty}\sup \|\hat x_i(t) - x_i(t)\| /  a_N 
$.
Although the result in Section III is given just for linear systems, the numerical results shown in the table show a  remarkable  improvement of the sensitivity to high-frequency measurement noise of the new observer with respect to the standard high-gain observer.\\

{\centering
{\footnotesize
\begin{tabularx}{\linewidth}{@{} lll @{} }
\toprule[1.5pt]
\multicolumn{1}{c}{\begin{tabular}[c]{@{}c@{}}Standard High Gain\\  Observer  $\hax$ \end{tabular}} 
& \multicolumn{1}{c}{\begin{tabular}[c]{@{}c@{}}Modified \\ Observer $\hax' = L_1 \xi$ \end{tabular}} 
& \multicolumn{1}{c}{\begin{tabular}[c]{@{}c@{}}Modified \\ Observer $\hax''= L_2 \xi$ \end{tabular}}  
\\\midrule
$\dnorm{\hax_1 - x_1}_a = 0.15$ 
& $\dnorm{\hax'_1 - x_1}_a = 0.06 $
& $\dnorm{\hax''_1 - x_1}_a = 0.06 $
\\ 
$\dnorm{\hax_2 - x_2}_a = 8$ 
& $\dnorm{\hax'_2 - x_2}_a = 0.2$
& $\dnorm{\hax''_2 - x_1}_a = 3 $
\\ 

$\dnorm{\hax_3 - x_3}_a = 2\cdot 10^2$ 
& $\dnorm{\hax_3' - x_3}_a =0.2$
& $\dnorm{\hax_3'' - x_3}_a =3$
\\
$\dnorm{\hax_4- x_4}_a = 2.5\cdot 10^3$
& $\dnorm{\hax_4' - x_4}_a =0.1$
& $\dnorm{\hax_4'' - x_4}_a =2$ 
\\
$\dnorm{\hax_5 - x_5}_a = 10^{4}$ 
& $\dnorm{\hax_5' - x_5}_a = 0.3$
& $\dnorm{\hax_5'' - x_5}_a = 0.3$
\\ 
\bottomrule[1.25pt]
\end {tabularx} \label{tab1} 

\par}
}{Table 1: Normalized asymptotic errors in presence of noise.}

\begin{center}
\begin{figure}[h]
\begin{center}\includegraphics[scale=.6]{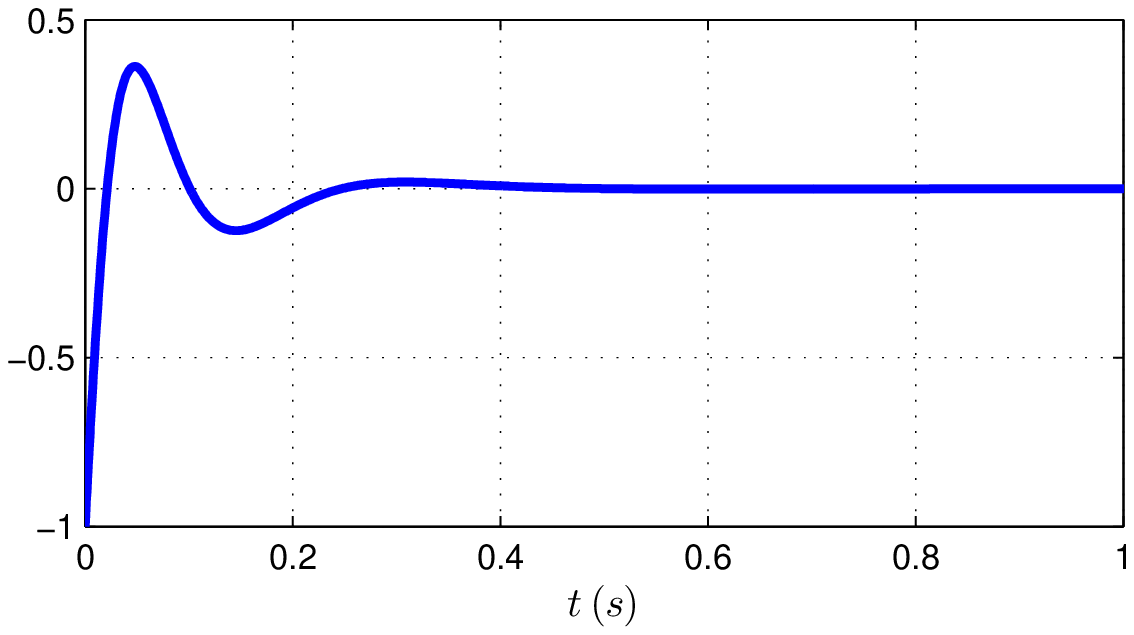}\end{center}
\caption{ Error state estimate $\norm{\hax''_1-x_1}$, $\norm{\hax'_1-x_1}$. }\label{graph1}
\end{figure}
\end{center}
\begin{center}
\begin{figure}[h]
\begin{center}\includegraphics[scale=.6]{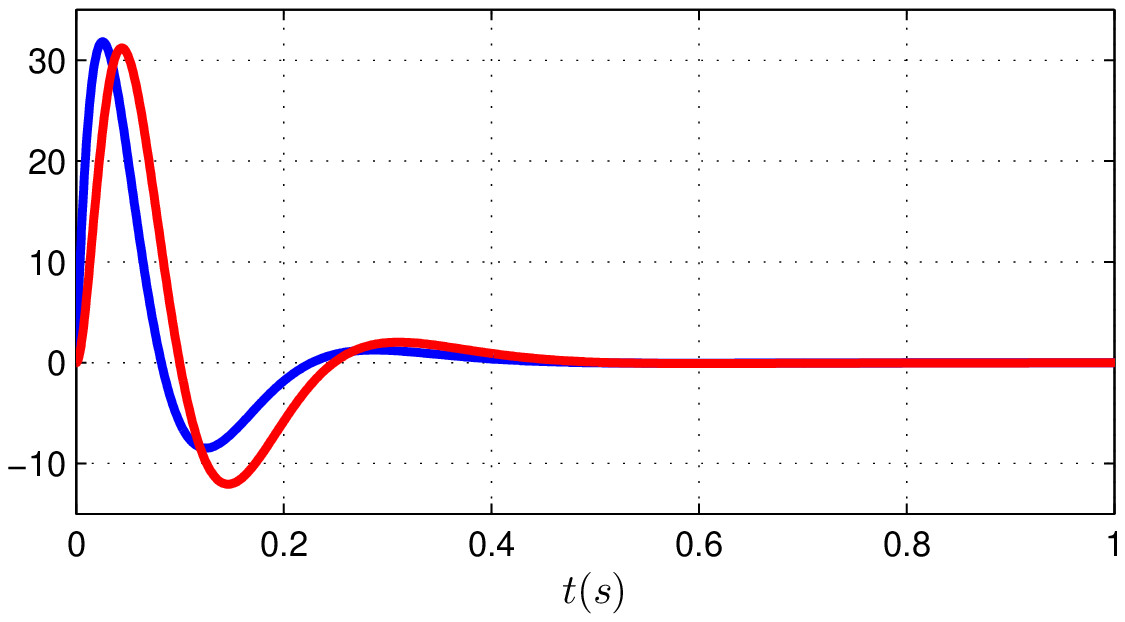}\end{center}
\caption{ Error state estimate  $\norm{\hax''_2-x_2}$ (blue line) and $\norm{\hax'_2-x_2}$ (red line).}\label{graph2}
\end{figure}
\end{center}

\section{Conclusions}
We presented a new observer design based on high-gain techniques
with a tunable state-estimate convergence speed.
With respect to standard high-gain observers the state dimension is larger ($2n-2$ instead of $n$) with 
a clear benefit in the observer implementation due to the power of the high-gain which is only $2$ and not $n$. 
Moreover, when specialised to linear systems, we showed the benefit of the proposed observer with respect to the standard high-gain observer
in terms of high-frequency noise rejection. Benefits that are clearly confirmed also  for the nonlinear Van-der Pol example numerically simulated in the previous section. A complete characterisation of the sensitivity to sensor noise of the new observer is an interesting research topic are that is now under investigation.
 
The  peaking phenomenon due to wrong initial conditions and fast convergence that is typical of high-gain observers is not prevented by our proposed structure. However, other techniques to deal with peaking (such as saturations, time-varying gains \cite{Sanfelice}, gradients techniques \cite{AstolfiCDC}, and others) are available and can be adopted to improve the proposed observer structure.
In this work we didn't consider the multi-output case. For the specific class of multi-output systems which are diffeomorphic to a block triangular form in which each block is associated to each output and it has a triangular dependence on the states of that subsystem (see  \cite{Atassi}), the  proposed structure can be simply applied block-wise to obtain a high-gain observer. Apart this case, a complete extension to the multi-output case is not immediate  and under investigation.\\

{\bf Acknowledgment.} We wish to thank Laurent Praly for suggesting the design procedure  presented in Appendix A.


\appendix

\subsection{Procedure to assign the eigenvalues of $M$}
Consider the matrices $M_i \in \RR^{2i-2} \times \RR^{2i-2}$ recursively defined as
\[
 M_1 = E_1\,, \qquad M_i = \left ( \ba{cc} M_{i-1} & \bar{N}_i\\ \bar{Q}_i & E_i \ea\right )\,, \qquad i=2,\ldots, n-1 
\]
where $\bar{N}_i= \mbox{col}(0_{2(i-2) \times 2}, \, N)$, 
$\bar{Q}_i= (0_{2(i-2) \times 2}, \, Q_i)$
and $E_i$, $i=1,\ldots,n-1$, $Q_i$, $i=2,\ldots, n-1$, and $N$ are defined as in the definition of $M$. Note that $M=M_{n-1}$ and, by letting $K_i=(k_{i1} \; k_{i2})^T$, note that $Q_i$ and $E_i$ depend on $K_i$, while $M_i$ depends on $K_1,\ldots, K_i$. We let ${\cal P}_{M_i}(\lambda)=\lambda^{2i} + m_1^i \lambda^{2 i -1} + \ldots + m^i_{2i-1} \lambda + m^i_{2i}$ and ${\cal P}_{M_{i-1}}(\lambda)=\lambda^{2i-2} + m_1^{i-1} \lambda^{2 i -3} + \ldots + m^{i-1}_{2i-3} \lambda + m^{i-1}_{2i-2}$ the characteristic polynomials of $M_i$ and $M_{i-1}$, and we use the notation $m^i_{[1,j]} = \mbox{col}(m_1^i, \ldots, m_j^i) \in \RR^{j}$,
$m^{i-1}_{[1,k]} = \mbox{col}(m_1^{i-1}, \ldots, m_k^{i-1}) \in \RR^{k}$ for some $j \leq 2i$ and $k\leq 2i-2$. 

The characteristic polynomial of ${\cal P}_{M_i}(\lambda)$ is computed as
\[\ba{rcl}
{\cal P}_{M_i}(\lambda) & = &  \lambda (\lambda + k_{i1}) {\cal P}_{M_{i-1}}(\lambda)
\\[.5em]
&& + k_{i2}\left[{\cal P}_{M_{i-1}}(\lambda) - \lambda(\lambda + k_{(i-1)1}){\cal P}_{M_{i-2}}(\lambda) \right]\,.
\ea\]

Hence, simple, although lengthy, computations show that the coefficients $m^i_{[1,2 i]}$ of ${\cal P}_{M_i}(\lambda)$ and $m^{i-1}_{[1,2i-2]}$ of ${\cal P}_{M_{i-1}}(\lambda)$  are related as follow
\beeq{\label{m_rec}
 \ba{rcl}
 m^i_{[1,2i-2]} &=& (I_{2i-2} + k_{i1} F) m^{i-1}_{[1,2i-2]} + k_{i1} \, {\rm v}_1\\[1mm]
 m^i_{2i-1} &=&  k_{i1}  \,m^{i-1}_{2i-2}  \\[1mm]
 m^i_{2i} &=& k_{i2} \,m^{i-1}_{2i-2}
 \ea
}
where ${\rm v}_1\in \RR^{2i-2}$ is the zero vector with a $1$ in the first position, and $F \in \RR^{(2i-2)\times (2i-2)}$ is the zero matrix with the identity matrix $I_{2i-3}$ in the lower left block. Note that $(I_{2i-2} + k_{i1} F)$ is invertible for all $k_{i1}$. Hence, from the first of (\ref{m_rec}), one obtains 
\[
m^{i-1}_{[1,2i-2]}  =\Lambda(m^{i}_{[1,2i-2]}, \, k_{i1})
\]  
where 
\[
\Lambda(m^{i}_{[1,2i-2]}, \, k_{i1}) =  (I_{2i-2} + k_{i1} F)^{-1} \, ( m^i_{[1,2i-2]} -  k_{i1} \, {\rm v}_1 \, )\,,
\]
which, embedded in the second and in the third of (\ref{m_rec}), yield the relations
\[
\sigma_1(m^i_{[1,2i-1]}, k_{i1})=0\,, \qquad k_{i2}=\sigma_2(m^i_{[1,2i]},\, k_{i1})
\]
where 
\[\ba{rcl}
 \sigma_1(m^i_{[1,2i-1]}, k_{i1})& = & k_{i1} \, {\rm v}_2^T \, \Lambda(m^{i}_{[1,2i-2]}, \, k_{i1}) -m^i_{2i-1}
\\[1em]
 \sigma_2(m^i_{[1,2i]},\, k_{i1})& = & \dst {m^i_{2i} \over {\rm v}_2^T \, \Lambda(m^{i}_{[1,2i-2]}, \, k_{i1})}
\ea\]
in which ${\rm v}_2 \in \RR^{2i-2}$ is the zero vector with a $1$ in the last position. We observe that $\sigma_1(\cdot,\cdot)$ is a polynomial in $k_{i1}$ of odd order $2i-1$. As a consequence, for any $m^i_{[1,2i-1]}$ there always exists at least one real $k_{i1}$ fulfilling  $\sigma_1(m^i_{[1,2i-1]}, k_{i1})=0$.

The previous results can be used to set up a "basic assignment algorithm" that is then used iteratively to solve the eigenvalues assignment of the matrix $M$.\\[1mm]
{\it Basic assignment algorithm}. Let $\bar {\cal P}_{i}(\lambda)=\lambda^{2i} + \bar m^i_1 \lambda^{2i-1} + \ldots + \bar m^i_{2i}$ be an arbitrary polynomial. Then, there exist a real $\bar K_i=(\bar k_{i1}, \bar k_{i2})^T$ and a polynomial $\bar {\cal P}_{i-1}(\lambda)=\lambda^{2i-1} + \bar m^{i-1}_1 \lambda^{2i-2} + \ldots + \bar m^{i-1}_{2i-2}$ such that
 \[
 \ba{rcl}
  K_i &=& \bar K_i\\ 
  {\cal P}_{M_{i-1}} &=& \bar {\cal P}_{i-1}(\lambda)
  \ea 
   \qquad \Rightarrow \qquad  
   {\cal P}_{M_{i}}(\lambda)=\bar {\cal P}_{i}(\lambda)\,.
 \]
As a matter of fact, by letting $\bar m^i_{[1,2i-1]}$ the coefficients of $\bar {\cal P}_{i}(\lambda)$, it is possible to take $\bar k_{i1}$ as a real solution of $\sigma_1(\bar m^i_{[1,2i-1]}, k_{i1})=0$,  $\bar k_{i2}=\sigma_2(\bar m^i_{[1,2i]},\, \bar k_{i1})$, and to take the coefficients $\bar m^{i-1}_{[1,2i-2]}$ of the polynomial $ \bar {\cal P}_{i-1}(\lambda)$ as   $\bar m^{i-1}_{[1,2i-2]} =\Lambda(\bar m^{i}_{[1,2i-2]}, \, \bar k_{i1})$ . $\triangleleft$\\[1mm]
With the previous algorithm in hand, the design of $K_1, \ldots, K_{n-1}$ to assign an arbitrary  characteristic polynomial to $M$, can be then immediately done by the following steps:

\begin{itemize}
 \item[1)] With $\bar {\cal P}_{n-1}(\lambda)$ the desired characteristic polynomial of $M$, compute $(\bar K_{n-1}, \bar {\cal P}_{n-2}(\lambda))$ by running the basic assignment algorithm with $i=n-1$.
 \item[2)] Compute iteratively $(\bar K_{i}, \bar {\cal P}_{i-1}(\lambda))$ by running the basic assignment algorithm for $i=n-2, \ldots, 2$.
 \item[3)] Compute $\bar K_1=(\bar k_{i1}, \bar k_{i2})^T$ so that $\lambda^2 + k_{i1} \lambda + k_{i2} = \bar {\cal P}_1(\lambda)$.
\end{itemize}

\end{document}